\documentclass[11pt]{article}

\usepackage[letterpaper]{geometry}
\usepackage{amsmath, amssymb, amsthm, thmtools, amsfonts}
\usepackage{color}
\usepackage{algpseudocode}
\usepackage{algorithm}
\usepackage{hyperref}
\usepackage{enumitem}
\usepackage{xspace}
\usepackage[utf8]{inputenc}
\usepackage[english]{babel}
\usepackage[T1]{fontenc}

\newtheorem{theorem}{Theorem}[section]
\newtheorem{lemma}[theorem]{Lemma}
\newtheorem{corollary}[theorem]{Corollary}
\newtheorem{remark}[theorem]{Remark}

\newtheorem{problem}[theorem]{Problem}

\newtheorem{definition}[theorem]{Definition}

\newcommand{\bigO}{\mathcal{O}}
\newcommand{\polylog}{\mathit{polylog}}
\newcommand{\set}[1]{\left\{#1\right\}}
\newcommand{\whp}{with high probability\xspace}
\newcommand{\prob}[1]{\mathsf{Pr}\left[#1 \right]}
\newcommand{\expv}[1]{\mathsf{E}\left[#1 \right]}
\newcommand{\paren}[1]{\mathopen{}\left(#1\right)\mathclose{}}
\newcommand{\ceil}[1]{\mathopen{}\left\lceil#1\right\rceil\mathclose{}}

\newcommand{\cclique}{$\mathsf{Congested}$\xspace$\mathsf{Clique}$\xspace}
\newcommand{\ncclique}{$\mathsf{Node}$\xspace$\mathsf{Congested}$\xspace$\mathsf{Clique}$\xspace}

\newcommand{\MPC}{$\mathsf{MPC}$\xspace}
\newcommand{\BSP}{$\mathsf{BSP}$\xspace}
\newcommand{\EP}{$\mathsf{EREW}$\xspace$\mathsf{PRAM}$\xspace}
\newcommand{\PRAM}{$\mathsf{PRAM}$\xspace}
\newcommand{\MR}{$\mathsf{MapReduce}$\xspace}

\newcommand{\SF}{$\mathsf{SF}$\xspace}
\newcommand{\MSF}{$\mathsf{MSF}$\xspace}

\renewcommand{\paragraph}[1]{\vspace{0.15cm}\noindent {\bf #1}:}

\definecolor{darkgreen}{rgb}{0,0.5,0}

\date{}

\author{
Krzysztof Nowicki\\
 \small University of Wroclaw  \\
 \small knowicki@cs.uni.wroc.pl
}

\begin{document}
\title{Random Sampling Applied to the MSF Problem \\
in the Node Congested Clique Model
    \thanks{This work was supported by the National Science Centre, Poland grant 2017/25/B/ST6/02010}
}

\maketitle              
    
\begin{abstract}
  The \cclique model, proposed by Lotker et al. \cite{Lotker_2005_MST}, was introduced in order to provide a simple abstraction for \emph{overlay networks}. \cclique is a model of distributed (or parallel) computing, in which there are $n$ players (nodes) with unique identifiers from set $\set{1,\dots, n}$, which perform computations in synchronous rounds. Each round consists of the phase of unlimited local computation and the communication phase. While communicating, each pair of nodes is allowed to exchange a single message of size $\bigO(\log n)$ bits. 
  
  Since in a single round each player can communicate even with $\Theta(n)$ other players, the model seems to be to powerful to imitate bandwidth restriction emerging from the underlying network. In this paper we study a restricted version of the \cclique model, the \ncclique model, proposed by Augustine et al. \cite{Augustine_2018_NCC}. The additional restriction is that in a single communication phase, a player is allowed to send / receive only $\bigO(\log n)$ messages.
  
  More precisely, we show that applying the algorithms from other models of distributed and parallel computing allows to provide an implementation of efficient subroutines in the \ncclique model. In  particular we provide communication primitives that improve the round complexity of the MST algorithm by Augustine et al. \cite{Augustine_2018_NCC} to $\bigO(\log^3 n)$ rounds, and give an $\bigO(\log^2 n)$ round algorithm solving the Spanning Forest (\SF) problem. Furthermore, we present an approach based on the random sampling technique by Karger, Klein and Tarjan \cite{Karger_1995_MST} that gives an $\bigO\paren{\log^2 n \ceil{\frac{\log \Delta}{\log \log n}}}$ round algorithm for the Minimum Spanning Forest (\MSF) problem.
  
  \noindent Besides the faster \SF / \MSF algorithms, we consider the key contributions to be:
  \begin{itemize}
    \item an efficient implementation of basic protocols in the \ncclique model,
    \item a tighter analysis of a special case of the sampling approach by Karger, Klein and Tarjan \cite{Karger_1995_MST} (and related results by Pemmaraju and Sardeshmukh \cite{Pemmaraju_2016_MST_o(m)}), 
    \item efficient $k$-sparse recovery data structure that requires $\bigO((k + \log n) \log n  \log k)$ bits and provides a recovery procedure that requires $\bigO((k + \log n) \log k)$ steps
  \end{itemize}
\end{abstract}


\renewcommand{\thefootnote}{\arabic{footnote}}

\setcounter{page}{0}
\thispagestyle{empty}
\newpage

\tableofcontents

\setcounter{page}{0}
\thispagestyle{empty}
\newpage

\section{Introduction}
\label{s:intro}

The fundamental theorem of software engineering [by David Wheeler] says that:

"We can solve any problem by introducing an extra level of indirection."\footnote{except for the problem of too many layers of indirection}

\noindent It is not a theorem in the mathematical sense, but rather a general principle for managing complexity through abstraction. Applying this principle to synchronous communication networks results in something we call an \emph{overlay network}. The abstraction is as follows: any two players can communicate in a single round and we do not take care about how this communication is realized by the underlying network. On the one hand it simplifies the process of algorithm design, on the other hand it deprives us of any control over how the communication is executed.

In the paper by Lotker et al. \cite{Lotker_2005_MST}, the authors introduce a simple abstraction for \emph{overlay networks}: the \cclique model. This is a model of distributed (or parallel) computing, in which there are $n$ players (nodes) with unique identifiers from set $\set{1,\dots, n}$ that perform computations in synchronous rounds. A single round consists of the phase of unlimited local computation and the communication phase. While communicating, each pair of nodes is allowed to exchange a single message of size $\bigO(\log n)$ bits.

The problem with modeling overlay networks via \cclique is that in a single round each player is allowed to communicate even with $\Theta(n)$ other players. This means that if an underlying network has some machines that have small number of physical connections, it may require really long time to simulate even a single round of the overlay network. 

To address this issue, in the paper by Augustine et al. \cite{Augustine_2018_NCC}, the authors propose a restricted version of the \cclique model, the \ncclique model. The additional restriction is that each player is allowed to send / receive only $\bigO(\log n)$ messages in a single round.

In this paper we continue studies initiated by Augustine et al. \cite{Augustine_2018_NCC}. There are multiple goals of our research -- the most important for us is to provide some basic tools and techniques for the \ncclique model. The secondary goal was to provide a faster algorithm for the Spanning Forest (\SF) and the Minimum Spanning Forest (\MSF) problems in the \ncclique. Finally, we want to extend the general understanding of the structure of the \SF and \MSF problems in decentralized models of computing.

More precisely, we show that it is possible to efficiently simulate the algorithms from other models of distributed and parallel computing, in order to provide efficient subroutines in the \ncclique model. Using this simulation, we provide some basic communication procedures, a sorting algorithm, and an implementation of edge-sampling and edge-recovery protocols relying on graph sketching \cite{Ahn_2012_sketches, Ghaffari_2016_MST_log*} and sparse recovery techniques \cite{Jowhari_2011_recovery, cormode_2014_sampling}.

Apart from the simulation theorems and the general purpose procedures, we provide algorithms for the Spanning Forest and the Minimum Spanning Forest problems in the \ncclique model. In the \SF problem, given a undirected graph, we have to identify a maximal acyclic subgraph. In the \MSF problem, given a undirected weighted graph, we want to identify a maximal acyclic subgraph that minimizes total weight of the edges.

Those problems are considered to be among the fundamental problems in the field of graph algorithmics. While, there are several approaches to the \SF and \MSF problems, we would like to recall the two results that are the most relevant to this paper: the first known \MSF algorithm (from 1926), by Otakar Borůvka \cite{Boruvka}, and the \MSF algorithm proposed by Karger, Klein and Tarjan \cite{Karger_1995_MST}. A general understanding of those results may be helpful to reading this paper.

While our approach to the \SF problem is basically an implementation of Borůvka's algorithm \cite{Boruvka}, the approach to the \MSF problems is slightly more complex -- it is based on the results on random sampling by Karger et al. \cite{Karger_1995_MST}, and the results of Pemmaraju et al. \cite{Pemmaraju_2016_MST_o(m)}. More precisely, the algorithm is a combination of the random sampling approach and the sparse recovery techniques. 

To provide an efficient implementation of this algorithm in the \ncclique model, we provide a tighter [than the result of Pemmaraju and Sardeshmukh \cite{Pemmaraju_2016_MST_o(m)}] analysis of the random sampling approach using $k$-wise independent random variables, when the probability of sampling is $1 / \Theta(k)$. Our results are rather incomparable to the results of Pemmaraju and Sardeshmukh \cite{Pemmaraju_2016_MST_o(m)} -- even though we achieve a better bound, in the general case we require a larger number of shared random bits.

\subsection{Related results / models}
In the paper introducing the \ncclique model \cite{Augustine_2018_NCC} the authors provide some connections between this model and \emph{hybrid networks} \cite{Gmyr_2017_Hybrid}, and \emph{$k$-machine model} \cite{Klauck_2015_KMachine}. We would like to point out that \ncclique, similarly as \cclique \cite{Lotker_2005_MST}, has connections to other \MPC (Modern / Massively Parallel Computing) \cite{Roughgarden_2016_MPC, Czumaj_2018_MPC, Ghaffari_2018_MPC} models like \MR \cite{Karloff_2010_MapReduce} and \BSP (Bulk Synchronous Parallel) \cite{Valiant_1990_BSP}. Furthermore, those connections extend to the older parallel computing models, e.g. Exclusive-Read Exclusive-Write Parallel RAM (\EP).

Even though several results have been published recently on the \MSF and Connected Components problems in the \MPC models (\cite{Lotker_2005_MST, Hegeman_2015_MST_logloglog, Ghaffari_2016_MST_log*, Korhonen_2016_MST_det, Pemmaraju_2016_MST_o(m), Jurdzinski_2018_MST_O(1)} in \cclique, \cite{Karloff_2010_MapReduce, Lattanzi_2011_filtering, Chitnis_2013_CC_MR, Kiveris_2014_CC_MR} in \MR, and \cite{Andoni_2018_CC_MPC} in $\mathsf{MPC}$), none of those results has a straightforward application in the \ncclique model. The main reason is that the \ncclique model has a small communication bandwidth, independent from the input size, while the \MPC models either have a large bandwidth by design or have a total bandwidth dependent on the number of edges in the input graph.

Still, the results from the \MPC and \PRAM models can be useful. If we manage to isolate some sufficiently small subproblems, such that the bandwidth is not a bottleneck anymore, we can try to apply some known results from other models to the \ncclique. In particular, in this paper, we use the results from the \cclique and \BSP models in order to provide some basic protocols of general use. Furthermore, we use the Connected Components and Spanning Forest algorithms \cite{halperin_1996_SF_PRAM} from the \EP model in order to solve instances of those problems that have sufficiently small number of edges properly distributed among the nodes of \ncclique.

\subsection{\ncclique}
\label{ss:ncclique_definition}
Formally the \ncclique model may be defined as a multi party communication model, in which:
\begin{itemize}
\item there are $n$ players (nodes)
\item each player has a unique identifier from set $\set{1, \dots, n}$
\item each player receives a part of the input set
\item computation is performed in synchronous rounds, each round consist of
\begin{itemize}
  \item a phase of local computation 
  \item a phase of communication
\end{itemize}
\item there is a communication link between each pair of players, however in the communication phase each node can send / receive only  $\bigO(\log n)$ messages of size $\bigO(\log n)$ bits
\end{itemize}

Although in the original version of the \ncclique model, the local computation is allowed to take arbitrarily long time to finish, we think that there should be a reasonable limit on the time complexity of local computation. In particular, all algorithms we provide have $\tilde \bigO(n)$ time complexity of each local computation step.

The goal of a \ncclique algorithm is to compute a value of some given function, defined on the sets of inputs of all players. As a result we can require either for all players to know the answer (which is reasonable only if a result is encodable on some small number of bits) or for each player to know some part of the result. In this paper we consider only the second variant.

\subsection{Spanning Forests in the \ncclique model}
For graph problems in the \ncclique we consider the following distribution of the input: each player corresponds to a single vertex of the input graph -- initially, the $i$th player knows all edges incident to the node with ID $i$ in the input graph. In particular this means that each edge is known by its two endpoints. 

In other words, we consider the \emph{vertex partition} multi party communication model: initially each player knows a set of all edges incident to a single vertex, as opposed to the \emph{edge partition} version in which each player initially knows an arbitrarily chosen set of edges. There is a significant difference between the \emph{vertex} and the \emph{edge} partitions, as for the \emph{edge partition} version there is an $\Omega(n^2 \log n)$ lower bound on the number of bits that have to be exchanged in order to compute \MSF \cite{Adler_1998_BSP_MTS}. Such lower bound implies an $\Omega(n / \log n)$ lower bound on the number of rounds, while for the \emph{vertex} partition there are $\polylog(n)$ round algorithms.

As a result, we require that for each edge from the resulting \SF / \MSF, at least one endpoint knows that this edge is a part of the \SF / \MSF. Again, requiring that both endpoints have to know all incident edges that are in the resulting \SF / \MSF, may require to deliver $\Omega(n)$ bits of information to a single node, which in turn implies an $\Omega(n / \log^2 n)$ round lower bound.

\begin{remark}
As in many other papers, we assume that in the graph of the \MSF problem no two edges have the same weight. The reduction from the general case to the case that no two edges have the same weight is by replacing the weight function $w$, by a function $w'$ defined as follows:$$w'\left(\langle u,v \rangle\right) = \left(w(\langle u,v\rangle),ID(u),ID(v)\right)$$.
\end{remark}

\section{High level description of the results}
This paper consists of three parts.

Firstly we introduce some communication primitives and discuss a simulation of the MPC / \EP algorithms in the \ncclique model. While our communication primitives can not do everything that is required by the algorithms by Augustine et al. \cite{Augustine_2018_NCC}, they are significantly faster in the special cases we consider. Furthermore, if we consider only the \MSF algorithm proposed by Augustine et al. \cite{Augustine_2018_NCC}, it also uses similar communication primitives, only in those special cases. Therefore, our efficient implementation yields a speed up of the \MSF algorithm by Augustine et al. \cite{Augustine_2018_NCC} by a $\Theta(\log n)$ factor.

Then, we discuss sparse recovery and graph sketching techniques \cite{Jowhari_2011_recovery, cormode_2014_sampling, Ahn_2012_sketches, Ghaffari_2016_MST_log*}, and provide efficient implementation of some edge recovery and edge sampling protocols in the \ncclique model.

In the last part of our paper, we provide algorithmic application of presented subroutines. As a warm up, we provide an implementation of a version of the Borůvka's algorithm, that finds a Spanning Forest of a graph in $\bigO(\log^2 n)$ rounds of \ncclique, \whp. Then we provide an alternative approach to the \MSF problem, which requires $\bigO\paren{\log^2 n \ceil{\frac{\log \Delta}{\log \log n}}}$ rounds, \whp. Moreover, since $\Delta \leq n$ the algorithm is faster at least by $\Theta(\log n \log \log n)$ factor than the \MSF algorithm by Augustine et al. \cite{Augustine_2018_NCC}. 

\subsection{Communication primitives}
The core of our efficient communication primitives is the sorting algorithm for the \BSP model proposed by Goodrich \cite{Goodrich_1999_Sorting}. Therefore, firstly we present a routing algorithm that allows us to simulate (some) BSP algorithms. Moreover, the routing procedure also gives an efficient simulation of (some) \EP algorithms, which we use in the later part of the paper. Then, we briefly discuss that one can efficiently simulate an instance of \ncclique with $\bigO(n)$ nodes on an instance of \ncclique with $n$ nodes in $\bigO(1)$ rounds. This simulation, together with the BSP sorting algorithm, allows us to provide efficient procedures propagating and aggregating information. More precisely, given a partition of nodes into disjoint sets, we provide procedures for:

\begin{itemize}
 \item propagating information among the members of a single set, simultaneously for all sets;
 \item aggregating information from the members of a single set in a distinguished member of this set, simultaneously for all sets.
\end{itemize}

\subsection{Sparse recovery and graph sketching techniques}
Our edge sampling and edge recovery protocols in the \ncclique model rely on the results of Jowhari et al. \cite{Jowhari_2011_recovery}, Cormode and Firmani \cite{cormode_2014_sampling}, Ahn et al. \cite{Ahn_2012_sketches} and Ghaffari and Parter \cite{Ghaffari_2016_MST_log*}. More precisely, for the edge sampling protocol we use existing $L_0$ sampler provided by Jowhari et al. \cite{Jowhari_2011_recovery} and Cormode and Firmani \cite{cormode_2014_sampling}, combined with graph neighbourhood encoding provided by Ahn et al. \cite{Ahn_2012_sketches}. 

For the edge recovery protocol, we use multiple instances of slightly modified version of this $L_0$ sampler. More precisely, we want to recover only vectors (that represent the sets of edges) that have the size of the support (the number of non zero coordinates) limited by some parameter $k$. With such promise we can use the $L_0$ sampler in a degree (or density) sensitive manner, as in the paper by Ghaffari and Parter \cite{Ghaffari_2016_MST_log*}.

\subsection{Spanning Forest Algorithms}
Our \SF algorithm is an implementation of Borůvka's algorithm \cite{Boruvka} adjusted to the \ncclique model. More precisely, to identify edges that are still between components we use graph sketching techniques \cite{Ahn_2012_sketches}. Furthermore, in order to merge connected components using those edges, we employ the algorithms for Connected Components and Spanning Forest from \EP \cite{halperin_1996_SF_PRAM}.

Our \MSF algorithm on the top level is based on the random sampling technique proposed by Karger, Klein, and Tarjan \cite{Karger_1995_MST} (KKT Sampling). 

The original version of the KKT Sampling Lemma says that we can reduce a single \MSF instance for a graph with $n$ nodes and $m$ edges to two sub-instances, such that for some $p \in (0,1)$:
\begin{itemize}
  \item the expected number of the edges in the \emph{first instance} is $mp$,
  \item the expected number of the edges in the \emph{second instance} is $n/p$.
\end{itemize}
Unfortunately, the \emph{second sub-instance} can be computed only after we find the \MSF $F$ for the \emph{first sub-instance}. This introduces some sequential dependence, which is unwelcome in distributed models of computation. Fortunately, the sequence of dependent instances (depth of recursion) in our case has length $\bigO\paren{\ceil{\frac{\log \Delta}{\log \log n}}}$. Moreover, each of those instances has a significantly smaller number of edges, which makes \ncclique computations easier.

For the \ncclique model we require a somehow stronger version of the KKT Sampling Lemma. The difficulties arise from the limited communication bandwidth. The former difficulty is that we have very limited shared randomness, therefore we execute the KKT Sampling with $k$-wise independent random variables, for some small value of $k$, instead of fully independent random variables. The latter difficulty is that naive implementation of the procedure computing the \emph{second instance} from the \emph{first instance} requires the \MSF of the \emph{first instance} to be known to every node in the network, which requires a lot of communication.

In the paper by Pemmaraju and Sardeshmukh \cite{Pemmaraju_2016_MST_o(m)}, the authors tackle similar issues: in particular they provide a procedure that computes the \emph{second instance} without widespread knowledge about $F$, which is correct even with the usage of random variables with bounded independence. However, we can not use their result as a building block of an efficient \MSF algorithm in the \ncclique model, as it still requires too large amount of communication.

Here, we provide an approach based on the results of Pemmaraju and Sardeshmukh \cite{Pemmaraju_2016_MST_o(m)}, exploiting the fact that we only consider some very special case of KKT Sampling. In particular we use for our purposes that $p^{-1} \in \bigO(\log n)$. This allows us to provide an efficient \ncclique implementation of the procedure finding $F$-light\footnote{The definition of the $F$-light edges appears in the later part of this paper. At this moment it is enough to know that such procedure is the main part of the procedure computing \emph{second instance} using the \MSF of the \emph{first instance}.} edges, relying on the linear graph sketching \cite{Ahn_2012_sketches} and the sparse recovery \cite{cormode_2014_sampling} techniques.

Besides the procedure identifying the set of $F$-light edges, we provide an implementation of the Borůvka's algorithm \cite{Boruvka} that provides the information we need for the procedure computing $F$-light edges. 

\section{Basic protocols}
\label{s:tools}
In this section we provide some basic procedures for the \ncclique model. 

Firstly we provide a deterministic routing protocol that relies on the results of Lenzen \cite{Lenzen_2013_Routing} for the \cclique model. Then we discuss various simulation results: we show that one can simulate larger instance of the \ncclique on a given instance; then, with use of the routing procedure, we show that the \ncclique model can simulate some special cases of the \BSP model (and other \MPC models), as well as the \EP model. As a corollary we get that we can efficiently sort in the \ncclique model using Goodrich's algorithm \cite{Goodrich_1999_Sorting} for the \BSP model. 

All those procedures allow us to efficiently solve slightly more complex communication problems: \emph{multicast} and \emph{aggregate}. In those problems, given a partition of nodes into disjoint set, in all sets simultaneously we want to either propagate some information or aggregate some information. More precise definition is provided later in this section. 

The \emph{multicast} and \emph{aggregate} protocols are based on well known approach -- for each set from a given division we build a balanced communication tree. The degree and depth of each tree is such that each level of a tree can communicate with a level directly below / above in $\bigO(1)$ rounds. Those communication trees are used by:
\begin{itemize}
\item the multicast protocol in order to propagate messages from the root to the leaves,
\item the aggregate protocol in order to aggregate in the root information from the leaves.
\end{itemize}

\subsection{Routing}
We define the \emph{routing problem} in the \ncclique in the following way.
\begin{problem}
Let $\alpha \in \Theta(\log n)$ be the number of different messages each node can send and receive in a single round of \ncclique. Each node is source and destination of $\bigO(\alpha)$ messages of size $\bigO(\log n)$. Initially only the sources know the destinations and contents of their messages. As a result, each node needs to learn all messages it is the destination of.
\end{problem}

\begin{theorem}
\label{t:routing}
The routing problem in the \ncclique can be solved in $\bigO(1)$ rounds, deterministically.
\end{theorem}

\begin{proof}
Initially, we split nodes into groups of size $\alpha$. in such a way that in the $j$th group there are nodes with IDs $\set{(j-1)\alpha + 1,\dots, j\alpha}$.

At the beginning we sort all messages by their destinations, in all group simultaneously. There are $\bigO(\alpha^2)$ messages, and we have a set of nodes which basically can communicate as a \cclique consisting of $\alpha$ nodes, therefore we can sort all those messages using the Lenzen's sorting algorithm \cite{Lenzen_2013_Routing} for the \cclique model, in $\bigO(1)$ rounds, deterministically. 

Then, each node sends the $i$th smallest message from local memory, to the  $i$th smallest member of its group (for $i > \alpha$, we send the message to $i'$th node, where $i' = ((i-1) \mod \alpha) +1$). At this point, each node in each group has only $\bigO(1)$ messages with the same destination. If that would not be the case, it would mean that the destination that occurs too many times is a destination of more than $\bigO(\alpha)$ messages, and that's forbidden by the problem statement. Since each node has some constant number of messages to each destination, we can send them one by one, in $\bigO(1)$ rounds.
\end{proof}

\subsection{Simulation}
In this subsection we briefly explain, that an \ncclique of size $n$ can simulate several different models of distributed and parallel computation. More precisely, we consider simulating larger instances of \ncclique model, as well as the \BSP and \EP models.

\paragraph{The \BSP and \EP models} 

In the \BSP model, similarly as in \ncclique, there are $n$ distinct machines (players / nodes), each has a unique identifier from set $\set{1,\dots, n}$, and each player initially knows a part of the input. The difference between \ncclique and \BSP is that in \BSP communication is performed in the following way: 
\begin{itemize}
\item players simultaneously exchange messages of size $\bigO(\log n)$ bits,
\item each pair of players can exchange multiple messages in a single round (still those messages are sent simultaneously by both players), 
\item the total number of messages sent / received by each machine is limited by some parameter $h$.
\end{itemize}

The \EP model, is a model of parallel computing, in which there are $p$ processors, with a random access to a shared memory of size $m$. The processors execute read / write operation exclusively, i.e. in a single step of the algorithm only one processors can access a single unit of memory. Moreover, in a single step each processor can access only $\bigO(1)$ different memory units.

\paragraph{Summary of the Simulation Theorems} 

In this section we show that it is possible to simulate:
\begin{itemize}
  \item a single round of $\bigO(n)$ node \ncclique in $\bigO(1)$ rounds of a \ncclique of size $n$ (Theorem \ref{t:simulation})
  \item a single round of the \BSP model with $\bigO(n)$ machines, and $h \in \bigO(\log n)$ in $\bigO(1)$ rounds of a \ncclique of size $n$ (Theorem \ref{t:simulation_BSP})
  \item a single round of \EP with $\bigO(n \log n)$ processors and $\bigO(n \log n)$ memory in $\bigO(1)$ rounds of a \ncclique of size $n$ (Theorem \ref{t:simulation_EP})
\end{itemize}

\newpage

\paragraph{Simulation of a larger instance of the \ncclique model}
\begin{problem}
  There are at most $cn$ nodes to be simulated on the $n$ node \ncclique. Node with ID $i$ knows the state of internal memory of all nodes with IDs that are congruent to $i$ modulo $n$. The goal is to perform communication between all simulated nodes, in $\bigO(1)$ rounds.
\end{problem}

\begin{theorem}
  \label{t:simulation}
  It is possible to simulate a \ncclique of size $\bigO(n)$ on a \ncclique of size $n$ in $\bigO(1)$ rounds.
\end{theorem}
\begin{proof}
The simulation is pretty much the same as in the paper \cite{Jurdzinski_2018_MST_O(1)}. In this proof, we call the nodes of the original \ncclique \emph{real}, and simulated nodes \emph{virtual}.  
The protocol is following: each \emph{real} node orders \emph{virtual} nodes assigned to it by ID. Thus, each \emph{real} node has a sequence of \emph{virtual} nodes of length up to $c$. 

Let consider a sequence of pairs $i,j$, such that $1 \leq i \leq j\leq c$. The protocol takes exactly as many phases as the length of this sequence. In the $k$th phase, we consider the $k$th element of this sequence. Let assume that this element is a pair $(a,b)$, then in the $k$th phase we perform communication between all \emph{virtual} nodes that are $a$th element of the local list of some \emph{real} node, and all \emph{virtual} nodes that are $b$th element on the local list of some \emph{real} node. This way, each pair of \emph{real} nodes has to exchange at most one message at a time.

The whole protocol requires $\Theta(c^2)$ rounds, and since $c$ is a constant, it requires $\bigO(1)$ rounds.

\end{proof}

\paragraph{Simulation of the BSP model}
\begin{problem}
  There are at most $cn$ \BSP machines to be simulated on the $n$ node \ncclique. Node with ID $i$ knows the state of internal memory of all machines with IDs that are congruent to $i$ modulo $n$. The communication bandwidth $h$ of a single machine is $\bigO(\log n)$. The goal is to perform communication between all simulated nodes, in $\bigO(1)$ rounds.
\end{problem}

\begin{theorem}
  \label{t:simulation_BSP}
  It is possible to simulate a single step of a \BSP model with $\bigO(n)$ machines, and $h \in \bigO(\log n)$, on a \ncclique of size $n$ in $\bigO(1)$ rounds.
\end{theorem}
\begin{proof}
We begin with the assignment of some virtual nodes to the nodes of \ncclique. More precisely, we assign to the node with ID $j$, which has to simulate $k$ \BSP machines, $k$ virtual nodes with IDs $j, j+n,  \dots, j+(k-1)\cdot n$. By Theorem \ref{t:simulation} we can simulate communication between all those nodes in $\bigO(1)$ rounds. Moreover, by Theorem \ref{t:routing}, we can execute routing scheme for those virtual nodes, as long as each virtual node sends / receives $\bigO(\log n)$ messages, in $\bigO(1)$ rounds. Therefore, we can simulate \BSP-like communication between all pairs of virtual nodes (thus simulated \BSP machines) in $\bigO(1)$ rounds of the \ncclique model of size $n$.
\end{proof}

As a corollary, we have that one can sort some set of keys in the \ncclique model, as long as it is sufficiently small, and evenly distributed. We define the \emph{sorting problem} in the \ncclique model in the following way.

\begin{problem}
Let $\alpha \in \bigO(\log n)$ be the number of different messages each node can send and receive in a single round of \ncclique. Each node is given $\bigO(\alpha)$ comparable keys of size $\bigO(\log n)$. Let $c$ be a value such that the total number of keys is no larger than $c \cdot \alpha n$. 

As a result, node $i$ needs to learn about the keys with indices $(i-1) \cdot c \alpha + 1,\dots, i \cdot c \alpha$ in a global enumeration of the keys that respects their order. Alternatively, we can require that nodes need to learn the indices of their keys in the total order of the union of all keys (i.e., all duplicate keys get the same index). 
\end{problem}

\begin{corollary}
\label{c:sorting}
It is possible to solve the sorting problem in \ncclique model, deterministically in $\bigO(\log n / \log \log n)$ rounds.
\end{corollary}

\begin{proof}
By Theorem \ref{t:simulation_BSP}, we have that the \ncclique model can simulate the \BSP model that has $n$ machines and $h \in \Theta(\log n)$. Since we want to sort $\bigO(n \log n)$ keys, we can use the sorting algorithm by Goodrich \cite{Goodrich_1999_Sorting} and solve the sorting problem in $\bigO(\log n / \log h) = \bigO(\log n / \log \log n)$ rounds, deterministically.
\end{proof}

\paragraph{Simulation of the \EP model} 

\begin{problem}
  There are at most $cn \log n$ \EP processors (with unique identifiers) and a shared memory of size $\bigO(n \log n)$ to be simulated on the \ncclique of size $n$. Node with ID $i$ knows the state of internal memory of all \EP processors with IDs that are congruent to $i$ modulo $n$, as well as the memory units that have addresses congruent to $i$ modulo $n$. The goal is to simulate read / write operations of all processors in $\bigO(1)$ rounds of the \ncclique.
\end{problem}

\begin{theorem}
  \label{t:simulation_EP}
  It is possible to simulate a single step of a \EP machine with $\bigO(n \log n)$ processors and $\bigO(n \log n)$ memory, on a \ncclique of size $n$ in $\bigO(1)$ rounds.
\end{theorem}

\begin{proof}
  By the definition of the problem, each \ncclique node has to simulate $\bigO(\log n)$ processors ans $\bigO(\log n)$ memory units. The assignment of processors / memory units is such that for each processor ID / memory address it is possible to determine which node of the \ncclique is responsible for simulating this particular processor / memory unit. 
  
  In order to execute \emph{write} operation, a node simulating processor executing \emph{write} sends information to be written to the node responsible for simulating target memory unit. Since each node of \ncclique simulates $\bigO(\log n)$ processors, and $\bigO(\log n)$ memory units, each node is a source and destination of $\bigO(\log n)$ different messages. Therefore, by Theorem \ref{t:routing}, such procedure can be executed in $\bigO(1)$ rounds. 
  
  The simulation of the \emph{read} operation can be done in a similar way. The difference between \emph{read} and \emph{write} is that whole \emph{write} is executed in two phases. In the first phase the node simulating processor sends to the node simulating target memory unit that it wants to read some particular memory unit. In the latter phase each node that received some \emph{read} requests sends requested data to the nodes that requested them. In both phases each node has to send or receive $\bigO(\log n)$ messages, therefore the whole procedure can be executed in $\bigO(1)$ rounds, by applying twice Theorem \ref{t:routing}.
 
\end{proof}

\noindent As a Corollary we get the following.
\begin{corollary}
\label{c:ep_sf}
It is possible to identify the Connected Components / a Spanning Forest of a graph $G$ represented as a list of $\bigO(n \log n)$ edges, which are evenly distributed among nodes of \ncclique, in $\bigO(\log n)$ rounds of \ncclique, \whp. As a result of the Connected Component algorithm, each node knows the ID of an arbitrary chosen member of its component. As a result of the Spanning Forest Algorithm, we get the evenly distributed list of edges that form a spanning forest.
\end{corollary}
\begin{proof}
  Since list of edges is evenly distributes, we have that each node of \ncclique knows $\bigO(\log n)$ edges. Therefore, we can treat such list of the edges as a initial state of memory of a \EP machine. Since, by the result of Halperin and Zwick \cite{halperin_1996_SF_PRAM}, we can identify the Connected Components or a Spanning Forest of a graph in $\bigO(\log n)$ rounds in \EP, by Theorem \ref{t:simulation_EP} we can identify the Connected Components or a Spanning Forest of $G$ in $\bigO(\log n)$ rounds of \ncclique. The results of those algorithms are as described in the Corollary \ref{c:ep_sf}.
  
\end{proof}

\subsection{Other Communication Primitives}
In this subsection we provide some efficient communication primitives for communication inside disjoint sets of nodes. In particular we provide the protocols solving the \emph{multicast} and \emph{aggregate} problems. More precisely we define the \emph{multicast} problem in the following way.
\begin{problem}
We are given a partition of nodes into $k$ disjoint sets $X_1, X_2, \dots, X_k$, and the set of nodes $v_1, v_2, \dots, v_k$. For each $i$, $v_i \in X_i$ and all members of set $X_i$ know the ID of $v_i$. The node $v_i$ has in local memory a message $m_i$ of size $c \leq \frac{\log^2 n}{2}$ bits. As a result, for all $i$ we require that all members of $X_i$ know $m_i$.
\end{problem}

\begin{theorem}
\label{t:multicast}
  It is possible to solve the multicast problem in $\bigO\paren{\frac{\log n}{2\log \log n - \log c}}$.
\end{theorem}

\begin{proof}
  The algorithm solving the multicast problem is as follows. 
  
  Firstly we sort all nodes by its 'special' node ID. After this, we have that nodes with the same 'special' node, are known by a single node or by a sequence of nodes with consecutive IDs. 
  
  Now, we can organize the nodes into a communication tree, starting from the construction of the leaves of this tree, and going up to the root. The leaves of this tree are all nodes that hold the result of the sorting algorithm. In order to create a higher level of the communication tree, at the highest existing level we do the following: 
  \begin{itemize}
  \item divide the nodes into groups of size $\Theta(\log n)$,
  \item in each group assign a node with the smallest ID to be a parent of all nodes in this group (in particular this node is also its own parent in this communication tree)
  \end{itemize}
  
  Since the degree of a internal node in this tree is $\Theta(\log n)$, whole communication tree has $\bigO(\log n / \log \log n)$ levels. Using this tree we can compute frequency of each 'special' ID, using simple dynamic programming in bottom-up fashion. Then, we can propagate the results to leaves in top-down fashion. Since depth of a tree is $\bigO(\log n / \log \log n)$, and any two levels can communicate in a $\bigO(1)$ rounds, whole counting procedure requires $\bigO(\log n / \log \log n)$ rounds.
  
  At this point each member of $X_i$ knows the size of each $X_i$. The idea is to build a separate communication tree for each disjoint set and special node. Since we can simulate a \ncclique with $\bigO(n)$ nodes in $\bigO(1)$ rounds, we can create some number of virtual nodes in order to construct those communication trees.
  
  Since we know the sizes of $X_i$, and the size of a message for the multicast problem, we can compute the number of internal nodes in a tree such that it's internal degree is as large as possible (to provide smaller depth), but a layer of the tree still can downcast the message to the lower level in $\bigO(1)$ rounds. The largest possible degree for a message of size $\bigO(c \log n)$, in order to use Theorem \ref{t:routing} for communication in $\bigO(1)$ rounds, is $\Theta\paren{\frac{\log^2 n}{c}}$. Therefore the depth of the tree is $\bigO\paren{\frac{\log n}{\log \frac{\log^2 n}{c}}} = \bigO\paren{\frac{\log n}{2\log \log n  - \log c}}$.
  
  Now we discuss, how to actually assign the virtual nodes to sets $X_i$ in such a way that we do not assign a single virtual node to multiple sets. To achieve this, we again use the communication tree already utilized for computing the sizes of sets $X_i$. In particular, for each 'special' ID, it is enough to know the number of virtual nodes assigned to all 'special' nodes with smaller ID. 
  
  Let say that for a node $v$ inside a communication tree, a $v$-prefix is the largest set of pairs $\langle$node-ID, special node ID$\rangle$, such that the largest member (in the sequence sorted by the second coordinate) of this set is in the leaf of the subtree of $v$. In order to achieve our goal, it is enough for each node $v$ to 
  \begin{itemize}
  \item compute the number of virtual nodes required for each rooted subtree of the communication tree (bottom up), 
  \item communicate to the $i$th child, the total number of virtual nodes in the prefix of $v$ decreased by number of virtual nodes in subtrees of children $\set{i, i+1, \dots}$(top down).
  \end{itemize}
  
  Since we have a set of virtual nodes assigned to each disjoint set, we can finally use them for building the communication trees. More precisely, we want $v_i$ as the root and members of $X_i$ as leaves, and assigned virtual nodes as internal nodes of the tree. Now we can finally use the tree for multicasting messages from $v_i$ to members of $X_i$, for all $i$ simultaneously. Since the depth of the tree is $\bigO\paren{\frac{\log n}{2\log \log n  - \log c}}$, Theorem \ref{t:multicast} is correct.
\end{proof}

Furthermore, if we execute the part in which we send the message from the root to all leaves in a \emph{pipelined fashion}, we can actually send multiple messages in the time longer by a constant factor. Since, while sending the message $m_i$ from root to leaves, we use only two consecutive levels of the communication tree at single moment, we can start sending an another message as soon as root and its children are free. Thus, we have the following Theorem:
\begin{theorem}
\label{t:pipelined_multicast}
  It is possible to solve the multicast problem for single partition into sets and $x$ messages per 'special' node in $\bigO\paren{x + \frac{\log n}{2\log \log n - \log c}}$.
\end{theorem}

\begin{corollary}
  It is possible to multicast $k$ bits message in 
  \begin{itemize}
    \item $\bigO\paren{\frac{k}{\log^2 n}}$ rounds of the \ncclique model, for $k \geq \log^3 n$,
    \item $\bigO\paren{\frac{\log n }{\epsilon \log \log n}}$ rounds of the \ncclique model, for $k = \frac{\log^{3-\epsilon}n}{\epsilon \log \log n}$,\\
    for $\frac{1}{\log \log n} < \epsilon < 2$.
  \end{itemize}
\end{corollary}

\begin{proof}
  For $k \geq \log^3 n$ we can split the $k$ bit message into $2k / \log^2 n$ parts of size $\frac{\log^2 n}{2}$, and use Theorem \ref{t:pipelined_multicast} with $x = \frac{k}{\log^2 n}$, and $c = \frac{\log^2 n}{2}$. For $k = \frac{\log^{3-\epsilon} n}{\epsilon \log \log n}$ and $\frac{1}{\log \log n} < \epsilon < 2$, we can split the message into $\Theta\paren{\frac{\log n}{\epsilon \log \log n}}$ parts of size $\log^{2-\epsilon} n$ and use Theorem \ref{t:pipelined_multicast} with $x = \frac{\log n}{\epsilon \log \log n}$, and $c \in \Theta(\log^{2-\epsilon} n)$
\end{proof}

\begin{corollary}
\label{c:shared_randomness}
  It is possible to provide shared randomness of size $\Theta(k)$ bits in 
  \begin{itemize}
    \item $\bigO\paren{\frac{k}{\log^2 n}}$ rounds of the \ncclique model, for $k \geq \log^3 n$,
    \item $\bigO\paren{\frac{\log n}{\epsilon \log \log n}}$ rounds of the \ncclique model, for $k = \frac{\log^{3-\epsilon} n}{\epsilon \log \log n}$,\\
    for $\frac{1}{\log \log n} < \epsilon < 2$.
  \end{itemize}
\end{corollary}

The \emph{aggregate} problem is complementary in some sense to the \emph{multicast} problem - instead of spreading information we have to gather information. More precisely, we define the \emph{aggregate} problem in the following way.

\begin{problem}
We are given a partition of nodes into $k$ disjoint sets $X_1, X_2, \dots, X_k$, and the set of nodes $v_1, v_2, \dots, v_k$. For each $i$, $v_i \in X_i$ and all members of set $X_i = \set{x^i_1, \dots, x^i_j}$ know the ID of $v_i$. For each $i$ and $j$, a node $x^i_j$ has some private input $y^i_j$. The goal is, for a given function $f$, to compute the value of $f\left(\set{y^i_1, \dots, y^i_j}\right)$, under assumption that $f$ has following properties:
\begin{itemize}
\item $f$ is an distributive aggregative function
\item value of $f$ can be encoded on $c \leq \frac{\log^2 n}{2}$ bits
\end{itemize}
As a result, for all sets we require that $v_i$ knows the value of $f\left(\set{y^i_1, \dots, y^i_j}\right)$.
\end{problem}

Where, by an \emph{distributive aggregative function} we understand exactly the same thing as the authors of \cite{Augustine_2018_NCC}:
\begin{definition}
The function $f$ is a distributive aggregative function when:
\begin{itemize}
  \item it maps a value of a multiset $S$ to some value $f(S)$
  \item exists function $g$, such that for any partition of the multiset $S$ into multisets $\set{S_1, \dots, S_l}$, $f(S) = g\left( \set{ f(S_1),f(S_2), \dots, f(S_l) }\right)$.
\end{itemize}
\end{definition}

\begin{theorem}
\label{t:aggregate}
  It is possible to solve aggregate problem in $\bigO\paren{\frac{\log n}{2\log \log n - \log c}}$.
\end{theorem}
\begin{proof}
  The algorithm for the multicast problem, is almost the same as for multicast - the difference is that we use constructed communication tree in order to combine values of $f$ bottom-up, instead of sending the message top-down.
\end{proof}

Furthermore, we can also define the pipelined version of the aggregate function, that is summarized by the following Theorem:
\begin{theorem}
\label{t:pipelined_aggregate}
  It is possible to solve the aggregate problem for a single partition into sets and $x$ functions per set in $\bigO\paren{x + \frac{\log n }{2\log \log n - \log c}}$.
\end{theorem}

\noindent Also, we have an analogous corollary. Let say that a distributive aggregative function $f$ that has a result encodable on $k$ bits is \emph{$\alpha$-splittable}, if there exists functions $f_1, f_2, \dots, f_{\ceil{k/\alpha}}$ such that the result of $f$ is a concatenation of the results of $f_1, f_2, \dots, f_{\ceil{k/\alpha}}$, and for each $i$, $f_i$ is distributive aggregative.
\begin{corollary}
  It is possible to aggregate a result of a $\alpha$-splittable distributive aggregative function, for $\alpha \leq \frac{\log^2 n}{2}$ that has result encoded on $k$ bits in
  \begin{itemize}
    \item $\bigO\paren{\frac{k}{\log^2 n}}$ rounds of the \ncclique model, for $k \geq \log^3 n$,
    \item for $k = \frac{\log^{3-\epsilon} n}{\epsilon \log \log n}$ and $\frac{1}{\log \log n} \leq \epsilon \leq 2$ in: \begin{itemize}
      \item $\bigO\paren{\frac{k}{\alpha}}$ rounds of the \ncclique, for $\alpha \geq \log^{2-\epsilon} n$,
      \item $\bigO\paren{\frac{\log n } {\epsilon \log \log n}}$ rounds  of the \ncclique, for $\alpha \leq \log^{2-\epsilon} n$.
      \end{itemize}
  \end{itemize}
\end{corollary}
\begin{proof}
For $k \geq \log^3 n$, since $f$ is $\alpha \leq \frac{\log^2 n}{2}$ splittable, we can split the result of $f$ into chunks of size $\frac{\log^2 n}{2}$. Then, we can use pipelined aggregate protocol, to aggregate each chunk as it is a separate distributive aggregative function. By Theorem \ref{t:pipelined_aggregate} it can be done in $\bigO\paren{\frac{k}{\log^2 n}}$ rounds.

For $k \leq \log^3 n$ the result depends on the parameter $\alpha$. If given function $f$ is $\alpha$ splittable, for $\alpha \leq \log^{2-\epsilon} n$, then we can split the result of $f$ into chunks of size $\bigO(\log^{2-\epsilon} n)$ bits, and by Theorem \ref{t:pipelined_aggregate}, with $x = \frac{k}{\log^{2-\epsilon} n}$, aggregate all chunks in $\bigO\paren{\frac{\log n}{\epsilon \log \log n}}$.
If given function $f$ is $\alpha$ splittable, for $\log^{2-\epsilon} n \leq \alpha \leq \frac{\log^2 n}{2}$,  then we can split the result of $f$ into chunks of size $\alpha$ bits, and by Theorem \ref{t:pipelined_aggregate} with $x = \frac{k}{\alpha}$ we aggregate all chunks in $\bigO\paren{\frac{k}{\alpha} + \frac{\log n}{2\log \log n - \log \alpha}} = \bigO\paren{\frac{k}{\alpha}}$.

\end{proof}

\subsection{Speed-up of the previously known algorithm}
In this subsection we provide a brief comment for the readers, who are familiar with the \MSF algorithm from the paper \cite{Augustine_2018_NCC}. 

The bottleneck of the algorithm seems to be that we require $\bigO(\log^2 n)$ random bits in order to generate $\log n$-wise independent hash functions. The authors provided communication primitives, which allowed multicasting among members of component $\Theta(\log^2 n)$ bits in $\bigO(\log n$) rounds, which implied that they can only generate $\bigO(1)$ hash functions per $\bigO(\log n)$ rounds. 

With our pipelined multicast procedure, some special node can announce in $\bigO(\log n)$ rounds $\Theta(\log^3 n)$ random bits. This allows for generating $\bigO(\log n)$ hash functions, that are $\log n$--wise independent, in $\bigO(\log n)$ rounds. Therefore, one can execute all sampling experiments simultaneously rather than sequentially. As far as we understand, this alone allows to solve the \MSF problem in $\bigO(\log^3 n)$ rounds in the \ncclique model.

\section{Edge sampling / recovery}
\label{s:recovery}
The algorithms we propose for the Spanning Forest and Minimum Spanning Forest problems rely on sparse recovery and graph sketching techniques \cite{Jowhari_2011_recovery, cormode_2014_sampling, Ahn_2012_sketches, Ghaffari_2016_MST_log*}. In this section we recall the results on graph sketching and sparse recovery, as well as provide some subroutines for edge sampling and recovery for the \ncclique model.

In subsection \ref{ss:knownresults} we recall the results of Jowhari et al. and Cormode and Firmani on $L_0$ sampling in the streaming model. In subsection \ref{ss:srecovery} we present application of computationally efficient $L_0$ sampler to the $k$-sparse recovery problem - we provide a protocol that requires $\bigO(\log^2 n \log \log n)$ bits of space and only $\bigO(\log n \ log \log n)$ time to recover all non zero coordinates, \whp. Then, in subsection \ref{ss:edgerecovery} , we recall a graph neighbourhood encoding by Ahn et al. and provide the algorithms for edge sampling (based on $L_0$ sampler) and $k$-edge recovery (based on $k$-sparse recovery) problems in the \ncclique model.

\subsection{Known results on $L_0$ sampling and $s$-sparse recovery}
\label{ss:knownresults}
We call a vector \emph{$s$-sparse}, when its support is no larger than $s$ (has at most $s$ non zero coordinates). The first lemma we recall concerns encoding and decoding $s$-sparse vectors using random linear functions. 
\begin{lemma} [Jowhari et al. in \cite{Jowhari_2011_recovery} claim that it is a well known result]\label{l:recovery}
\strut\\
For $1 \leq s \leq n$ and $k \in O(s)$ there is a random linear function\footnote{I feel like there is a silent assumption in \cite{Jowhari_2011_recovery} that by $\mathbb{R}$ we still mean real numbers encodable on $\bigO(\log n)$ bits, but sufficiently long to behave like real real numbers w.h.p. Even though there are no comments on this matter, results suggests that's the case} $L : \mathbb{R}^n \rightarrow \mathbb{R}^k$ (generated from $\bigO(k \log n)$ random bits ) and a recovery procedure that on input $L(x)$ outputs $x' \in \mathbb{R}^n$ or DENSE, satisfying that for any $s$-sparse $x$ the output is $x' = x$ with probability $1$, otherwise the output is DENSE
\whp.
\end{lemma}

The sparse recovery technique from Lemma \ref{l:recovery} was used by Jowhari et al. as a building block of an $L_0$ sampler of size $\bigO(\log^2 n)$, that has a constant probability of success \cite{Jowhari_2011_recovery}. Ahn et al. use exactly this result in the paper introducing the graph sketching techniques for the streaming model \cite{Ahn_2012_sketches}. To our understanding the drawback of using the $s$-sparse recovery provided by Lemma \ref{l:recovery} is that to recover an $s$-sparse $n$-dimensional vector it requires time $n^{\Theta(s)}$.

In the paper by Cormode and Firmani \cite{cormode_2014_sampling} authors provide alternative $s$-sparse recovery technique that can be plugged into linear $L_0$ sampler construction. Their recovery algorithm requires $\bigO(s \log n \log \frac{s}{\delta})$ bits of space, and recovers $s$-sparse vectors with probability $1-\delta$ in time $\bigO(s \log \frac{s}{\delta})$. On the one hand it allows to construct $L_0$ sampler with a constant probability of success that requires only $\bigO(\log^2 n)$ bits of space (which matches the result of Jowhari et al. with perfect $s$-sparse recovery from Lemma \ref{l:recovery}), on the other hand if we want to use it directly to recover $\Theta(\log n)$-sparse vectors, \whp, we would need $\Theta(\log^3 n)$ bits. 

For our purpose, we use the version of $s$-sparse recovery provided by Cormode and Firmani \cite{cormode_2014_sampling}, as it provides linear $L_0$ sampler in $\bigO(\log^2 n)$ bits of space, with constant probability of success and reasonable time complexity of recovery. 

\begin{lemma} \cite{cormode_2014_sampling}
\label{l:sampling_eff}
There exists a linear random sampling data structure that returns random non zero coordinate of an $n$-dimensional vector with constant probability, and if it fails to recover the coordinate, with high probability returns FAIL. The data structure requires $\bigO(\log^2 n)$ bits of space, and the encoding / decoding protocols require $\bigO(\log n)$ time.
\end{lemma}

Moreover, rather than directly apply the $s$-sparse recovery result of Cormode and Firmani \cite{cormode_2014_sampling}, we can use their version of $L_0$ sampler as a building block of $\bigO(\log n)$-sparse recovery data structure that requires only $\bigO(\log^2 n \log \log n)$ bits of space, succeeds with high probability and has a reasonable time complexity of recovery procedure. 

In other words, the $L_0$ sampler we use is as good as one provided by Jowhari et al. \cite{Jowhari_2011_recovery} and our $\bigO(\log n)$-sparse recovery data structure requires space larger only by $\Theta(\log \log n)$ factor than the one by Lemma \ref{l:recovery}, and both provide $\polylog(n)$ time sampling / recovery procedures.

\subsection{Efficient $k$-sparse recovery}
\label{ss:srecovery}
The approach is loosely based on degree sensitive sketching \cite{Ghaffari_2016_MST_log*}, the implementation of $k$-edge connectivity algorithm in the streaming model \cite{Ahn_2012_sketches}, and the $s$-sparse recovery for $s \in \bigO(1)$ \cite{cormode_2014_sampling}, and requires only $\bigO(\log^2 n \log k + k \log k \log n)$ bits of space, for recovering of $k$-sparse vectors.

\begin{lemma}
\label{l:recovery_eff}
There exists $k$-sparse recovery linear randomized data structure that
\begin{itemize}
  \item requires $\bigO(\log^2 n \log k + k \log k \log n)$ random bits for construction,
  \item has size $\bigO(\log^2 n \log k + k \log k \log n)$ bits, 
  \item allows to recover all non zero coordinates of $k$-sparse vector, \whp, in time $\bigO(\log n \log k + k \log k)$.
\end{itemize}
\end{lemma}

\begin{proof}
The $L_0$ sampler by Jowhari et al. \cite{Jowhari_2011_recovery} consists of $\Theta(\log n)$ instances of $s$-sparse recovery data structures. The $i$th instance is responsible for sampling a random non zero coordinate from a vector that has support $\Theta(2^i)$. Therefore, for a vector with support $k$, with constant probability, there exists an instance of $s$-sparse recovery data structure ($(\log k)$th instance) that allows to identify random non zero coordinate of this vector.

Therefore, if we have a guarantee that the vector has at most $k$ non zero coordinates, in order to identify a random non zero coordinate, with constant probability, we need only $\bigO(\log k)$ instances of $s$-sparse recovery data structure, for $s\in \Theta(1)$. Same idea was used by Ghaffari and Parter \cite{Ghaffari_2016_MST_log*}, however we rather use implementation of $L_0$ sampler instead of using an alternative implementation of graph sketching techniques. We call an $L_0$ sampler that works for vectors with $\Theta(k)$ supports, a $k$-sparse samplers.

Since the implementation by Cormode and Firmani \cite{cormode_2014_sampling} is basically $L_0$ sampler by Jowhari et al. \cite{Jowhari_2011_recovery}, with alternative $s$-sparse recovery data structure, we can use their version. Since single $\Theta(1)$-sparse recovery data structure that has a constant probability of success requires $\Theta\paren{\Theta(1) \log n \log\paren{\Theta(1) / \Theta(1)}}$ bits, whole structure that returns a random non zero coordinate of a $k$-sparse vector requires only $\bigO(\log n \log k)$ bits. Moreover, since the implementation by Cormode and Firmani \cite{cormode_2014_sampling} uses random variables with bounded dependence, whole data structure requires only $\bigO(\log n \log k)$ shared random bits.

In order to construct $k$-sparse recovery data structure that succeeds \whp, we can use $\Theta(k + \log n)$ instances of $k$-sparse sampler. To each such sampler we attach random bits used for computing this sampler. Since it is only $\bigO(\log n \log k)$ bits, the size of the structure increases only by a constant factor. 

The recovery procedure is loosely based on the $k$-edge connectivity from the streaming model \cite{Ahn_2012_sketches}. More precisely, we use each of $k$-sparse samplers in order to get a non zero coordinate of the vector we want to recover. If it succeeds, we subtract recovered coordinate from the samplers we did not used yet for recovery. We can do that, as each sampler has attached the random bits that were used for its generation. 

More precisely, if we sampled $j$th coordinate, $v_j$, we do the following. For each each unused sampler $S$ we compute a sampler $S(j, v_j)$ that uses the same random bits, but for vector that has $v_j$ as $j$th coordinate and zeros as all other coordinates. Then, since the construction of samplers is linear, we can use $S - S(j, v_j)$ instead of $S$ in later sampling attempts. This ensures that all unused samplers can not sample the $j$th coordinate again, thus no coordinate can be sampled twice.

Moreover, if we start with vector that is $k$ sparse, deletion of some coordinates can not make this vector have more than $k$ non zero coordinates. Since no coordinate can be sampled twice, in order to recover $k' \leq k$ coordinates we only need that $k' \leq k$ samplers succeeded. Since we have $\Theta(k + \log n)$ independent instances of $k$-sparse sampler, by standard Chernoff bound can recover all $k' \leq k$ coordinates, \whp.
\end{proof}

\subsection{Graph sketching and $k$-edge recovery}
\label{ss:edgerecovery}
By the graph neighbourhood of a set of nodes $C \subseteq V$ in graph $G = (V, E)$, we understand set of edges $E_C \subseteq E$, such that the edges in $E_C$ have exactly one endpoint in the set $C$. In the paper by Ahn et al. \cite{Ahn_2012_sketches} the authors provide for a set of nodes $C$ an encoding $S(C)$ of a graph neighbourhood of $C$, such that
\begin{itemize}
\item $S(C)$ is a $\Theta(n^2)$ dimensional vector, 
\item $S$ is linear: $S(C_1) + S(C_2) = S(C_1 \cup C_2)$ for any sets of nodes $C_1$, $C_2$.
\end{itemize}

\noindent For a node $v$, in the graph $G = (V,E)$, its neighbourhood is represented by a $n^2$ dimensional vector $S(\set{v})$ with coordinates $S(\set{v})_{\langle j,k \rangle}$ for $\langle j,k \rangle \in \set{1,2,\dots, n}^2$, where
\begin{itemize}
  \item $S(\set{v})_{\langle j,k \rangle} = 1$ iff $i = j$, and $\set{v_j, v_k} \in E$
  \item $S(\set{v})_{\langle j,k \rangle} = -1$ iff $i = k$, and $\set{v_j, v_k} \in E$
  \item $S(\set{v})_{\langle j,k \rangle} = 0$ otherwise
\end{itemize}
Moreover Ahn et al. provide the following Lemma:
\begin{lemma}\cite{Ahn_2012_sketches}
\label{l:from_sketches}
Let $E_C = E(C, V \setminus C)$ be the set of edges across the cut $(C, V \setminus C)$. Then, $|E_C|$ is equal to the half of the number of non zero coordinates of vector $x = \sum\limits_{v \in C} S(\set{v})$. Furthermore, $|x_{\langle j,k \rangle}| > 0$ iff $\set{v_j,v_k} \in E_C$.
\end{lemma}

Now we can finally provide a protocol solving the \emph{$k$ edge recovery} problem, defined as follows:

\begin{problem}
We are given a graph and a partition of nodes into $c$ disjoint sets $X_1, X_2, \dots, X_c$ and the set of nodes $v_1, v_2, \dots, v_k$. For each $i$, $v_i \in X_i$ and all members of set $X_i$ know the ID of $v_i$. We have guarantee that each set has at most $k$ outgoing edges, that is $\forall i, |E_{X_i}| \leq k$, where = $E_{X_i} = \set{\set{u,v}\ |\ u \in X_i \wedge v \in X_j \text{ for }j \not = i }$. As a result we require that for each $i$, $v_i$ knows $E_{X_i}$.
\end{problem}

\begin{theorem}
\label{t:k_recovery}
It is possible to solve the $k$ edge recovery problem in 
\begin{itemize}
\item $\bigO(k \log k / \log n)$ rounds, \whp, for $k \geq \log^2 n / \log \log n$, 
\item $\bigO(\log n / (\epsilon \log \log n))$ rounds, \whp, for $k = \log^{2-\epsilon} n / ((2-\epsilon)\epsilon \log \log^2 n)$. 
\end{itemize}
\end{theorem}

\begin{proof}
The protocol for the $k$ edge recovery is following. At the beginning, node with id $1$ generates $\bigO(k \log n + \log^2 n)$ random bits. Then, by Corollary \ref{c:shared_randomness}, those bits can be announced to all nodes in
\begin{itemize}
\item $\bigO(k \log k / \log n)$ rounds, for $k \geq \log^2 n / \log \log n$, 
\item $\bigO(\log n / (\epsilon \log \log n))$ rounds, for $k = \log^{2-\epsilon} n / ((2-\epsilon)\epsilon \log \log^2 n)$. 
\end{itemize}
At the beginning each node encodes its neighbourhood using encoding by Ahn et al. \cite{Ahn_2012_sketches}. Then, with shared random bits, each node can generate $\Theta(k + \log n)$ instances of linear $k$-sparse recovery data structure from Lemma \ref{l:recovery_eff}. Each of those instances can be treated as a set of $\Theta(\log k)$ independent $\Theta(1)$-sparse recovery data structures from the paper \cite{cormode_2014_sampling}. Each of them is a linear function, therefore it is also distributive aggregative. Total size of $k$-sparse recovery data structure is $\bigO((k +\log n)\log k \log n)$ bits, therefore in each set $X_i$ we can aggregate $k$-sparse recovery data structures from its members in node $v_i$, by Theorem \ref{t:pipelined_aggregate}, in 
\begin{itemize}
\item $\bigO(k \log k / \log n)$ rounds, for $k \geq \log^2 n / \log \log n$, 
\item $\bigO(\log n / (\epsilon \log \log n))$ rounds, for $k = \log^{2-\epsilon} n / ((2-\epsilon)\epsilon \log \log^2 n)$. 
\end{itemize}

By Theorem \ref{l:recovery_eff}, the node $v_i$ can decode the vector as long as it has $\bigO(k)$ non zero coordinates. In the problem definition we have an assumption that for each $i$, $|E_{X_i}|\leq k$. Therefore, by Lemma \ref{l:from_sketches}, $S(X_i)$ has $\bigO(k)$ non zero coordinates. Thus, the vector $S(X_i)$ can be decoded by $v_i$, using aggregated information. Therefore all nodes $v_i$ know the set $E_{X_i}$, and the number of rounds required for the protocol to finish is as stated in Theorem \ref{t:k_recovery}.
\end{proof}

\section{Spanning Forests Algorithms}
\label{s:mst}
With tools provided in Section \ref{s:tools} and Section \ref{s:recovery}, we can proceed to the algorithms for the Spanning Forest and  Minimum Spanning Forest problems. 

In the former part of this section we present the \SF algorithm that is an implementation of slightly modified version of the Borůvka's algorithm \cite{Boruvka}. In this algorithm, the computation is performed in phases. The invariant of each phase is that the algorithm maintains some spanning forest (starting with each node as a separate tree), such that it is a subgraph of the final minimum spanning forest. In each phase, we select for each tree (connected component) in this forest a \emph{minimum weight outgoing edge} (MWOE). In each phase, the algorithm uses MWOEs to merge trees inside the maintained spanning forest. After $\ceil{\log n}$ phases, the maintained forest is the minimum spanning forest of the input graph. One of modifications for the \SF problem is that, since we do not care about weights, we can simply select an arbitrary outgoing edge for each component in each phase.

In the latter part of this section we provide the \MSF algorithm that on the top level is based on the random sampling approach by Karger, Klein and Tarjan \cite{Karger_1995_MST}. Let consider an $n$ node graph $G$, and $F$ be the Minimum Spanning Forest of a random subgraph of $G$. By the \emph{$F$-light} edges in $G$ we understand those edges $e \in G$ that are not the heaviest edge on a cycle in $F \cup \set{e}$. They propose the following Lemma:
\begin{lemma}\cite{Karger_1995_MST}
\label{l:KKT}
Let $H$ be a subgraph obtained from $G$ by including each edge independently with probability $p$, and let $F$ be the minimum spanning forest of $H$. The expected number of $F$-light edges in $G$ is at most $n/p$ where $n$ is the number of vertices of $G$.
\end{lemma}
 Moreover, the \MSF of $G$ consists only of edges that are either in the forest $F$ or are $F$-light. Therefore, using Lemma \ref{l:KKT}, we can reduce computing \MSF of some dense graph, to 
\begin{itemize}
  \item computing $F$, the \MSF of the subgraph $H$,
  \item identifying $F$-light edges in $G$,
  \item computing the \MSF of $F$ extended by the set $F$-light edges.
\end{itemize}

\subsection{Spanning Forest Algorithm}
\label{s:sf}
In this subsection we show the following.
\begin{theorem}
\label{t:sf}
For a given unweighted, undirected, $n$ node graph $G$, it is possible to find a Spanning Forest of $G$ in the $\bigO(\log^2 n)$ rounds of the \ncclique model, \whp.
\end{theorem}

\begin{proof}

The algorithm is an implementation of slightly modified version of the Borůvka's algorithm. As in the original version we maintain a spanning forest of the input graph, and in phases we merge trees in this forest. The differences are that
\begin{itemize}
\item we consider unweighted graphs
\item we identify an outgoing edge for a component with constant probability
\end{itemize}
The first difference between our version and the one provided by Borůvka \cite{Boruvka} allows us to chose arbitrary outgoing edge rather than the one with the minimum weight. As for the second alteration, it does not increase the number of required phases (similar analysis as in \cite{Ghaffari_2017_mixing, Ghaffari_2016_planar}), and together with the first modification allows us to use graph sketches that work with constant probability.

\paragraph{Phase Invariants}
At the beginning of each phase we have following invariants:
\begin{itemize}
\item for each edge that were used for merging in all previous phases, at least one endpoint is aware of that
\item each node knows an ID of a leader (some arbitrary assigned member) of its connected component in current spanning forest
\end{itemize}

At the beginning of the algorithm, there are no edges that was used for merging in the previous phases, each node is in separate connected component and each node is a leader of its component. Thus, all invariants hold.

\paragraph{Single Phase}
Since in a current spanning forest each node knows ID of a leader of its component, we can use graph sketching approach (Lemma \ref{l:sampling_eff} and graph neighbourhood encoding by Ahn et al. \cite{Ahn_2012_sketches}) and pipelined aggregate protocol, in order to identify with constant probability an outgoing edge with constant probability in $\bigO(\log n / \log \log n)$ rounds. At this point, we have a set of edges between nodes. In order to identify the new division into connected components, we do the following:

\begin{itemize}
\item for each edge identify component ID's of two endpoints
\item run \EP Connected Components / \SF algorithm
\item each leader in the old division announces its component ID in the new division to all nodes in its component in the old division
\end{itemize}

The first and the third step can be done in $\bigO(\log n / \log \log n)$ rounds, by Theorem \ref{t:multicast}. As for the second step, we can do the following: To each edge, we attach the information with ID of a leader that sampled this edge. Since we have only $\bigO(n)$ edges, and each is held by some different node, we can identify Connected Components and Spanning Forest in the component graph via \EP algorithm (Corollary \ref{c:ep_sf}) in $\bigO(\log n)$ rounds, \whp. Since each edge has attached ID of a leader that sampled it, we can notify this leader that the edge is in the Spanning Forest in $\bigO(1)$ rounds, by Theorem \ref{t:routing}. Then, all leaders simultaneously in all components of old division can announce what was the sampled edge and whether it is in the resulting Spanning Forest, in $\bigO(\log n / \log \log n)$ rounds, by Theorem \ref{t:multicast}.

\paragraph{Summary of the Spanning Forest Algorithm}
Since we can implement a single phase of the algorithm in $\bigO(\log n)$ rounds, \whp, and by the results of \cite{Ghaffari_2016_planar, Ghaffari_2017_mixing} algorithm requires $\bigO(\log n)$ rounds, total number of rounds required by our algorithm is $\bigO(\log^2 n)$, \whp.
\end{proof}

\subsection{Minimum Spanning Forest Algorithm}
In the remaining part of this section, we provide a Minimum Spanning Forest algorithm. More precisely, we prove the following:
\begin{theorem}
\label{t:mst}
For a given weighted, undirected, $n$ node graph $G$, it is possible to find the Minimum Spanning Forest of $G$ in the $\bigO\paren{\log^2 n \ceil{\frac{\log \Delta}{\log \log n}}}$ rounds of the \ncclique model, \whp.
\end{theorem}
\begin{proof}
On the top level our \MSF algorithm uses an approach based on recursive application of KKT Sampling Lemma \cite{Karger_1995_MST}. A straightforward application of Lemma \ref{l:KKT} in \cclique model \cite{Hegeman_2015_MST_logloglog, Ghaffari_2016_MST_log*, Jurdzinski_2018_MST_O(1)} required that the whole forest $F$ is known by all nodes in order to identify $F$-light edges. In \ncclique even a single node can not know the forest $F$, simply because of the bandwidth restriction.

In the paper by Pemmaraju and Sardeshmukh \cite{Pemmaraju_2016_MST_o(m)}, the authors provide a slightly different approach to the KKT Sampling that does not require any centralized knowledge of the whole $F$. More precisely they use the centralized knowledge in order to coordinate communication, but if we could coordinate communication in some other way, there is no need of gathering information about all the edges of $F$ in a single node.

Furthermore, they provide an analysis, which gives so called component--wise bound on the number of the $F$-light edges. Roughly speaking, if we run Borůvka's algorithm on $F$, and look at the set of connected components in some fixed phase, the number of $F$-light edges in $G$ outgoing from each component that are lighter than MWOE of this component is bounded. Then, they show that the set of such edges contains all $F$-light edges, and provide a procedure identifying them, via graph sketching techniques \cite{Ahn_2012_sketches, Hegeman_2015_MST_logloglog}. Thus, they provide an algorithm that requires in total to send the number of message that is only $\polylog(n)$ factor larger than the size of the set of $F$-light edges.

Unfortunately, this result still requires too large communication to use it in the \ncclique model without any adjustments. In this section we provide a version of this approach suited for the \ncclique model.

\subsection{Top level description}
As mentioned before, the algorithm is a recursive application of the KKT Sampling Lemma \cite{Karger_1995_MST}. More precisely, given a graph of degree bounded by $\Delta$, we select a random subgraph, by including each edge with probability $p = 1 / \log n$. In order to execute sampling step, we use $\bigO(\log n)$-wise independent hashing functions. As long as both endpoints of an edge know the same hashing functions, this allows for both endpoints to make consistent decisions. Moreover, in order to generate $\Theta(\log n)$ independent hashing functions that are known to all nodes, we need a shared randomness of size $\Theta(\log^2 n)$. Since we sampled each edge with probability $p$, we get that the first instance has degree bounded by $\bigO(p \Delta + \log n)$, with high probability. 

The general idea is to provide an algorithm for graphs with $\bigO(n \log n)$ evenly distributed edges ($\bigO(\log n)$ edges per node), and use KKT sampling with probability $p$ as long as degree of a graph is not $\bigO(\log n)$. Such graph is represented as a list of evenly distributed $\bigO(n \log n)$ edges.  

Since, given a graph of degree $\Delta$, single application of KKT Sampling Lemma with probability of sampling $p=1 / \log n$ gives a graph with degree smaller by a factor $p$, we need up to $\bigO\paren{\ceil{\log_{p^{-1}}\Delta}} = \bigO\paren{\ceil{\frac{\log \Delta}{\log \log n}}}$ recursive calls in order to bring down the degree of a graph to $\bigO(\log n)$. Therefore, total number of shared random bits required for sampling on each of $\bigO\paren{\ceil{\frac{\log \Delta}{\log \log n}}}$ levels of recursion is $\bigO\paren{\log^2 n \ceil{\frac{\log \Delta}{ \log \log n}}}$. By Corollary \ref{c:shared_randomness} we can provide those random bits in $\bigO(\log n)$ rounds.

To complete the algorithm, we provide an algorithm that computes \MSF $F$ of a graph represented by $\bigO(n \log n)$ evenly distributed edges in a way that provides information useful for generation of the $F$-light edges. Also, as a result for each edge in \MSF at least one endpoint knows that this edge is in the computed \MSF. Then we provide the algorithm that recovers the $F$-light edges given the result of aforementioned algorithm. Finally, we provide how to reduce the forest $F$ extended by the $F$-light edges into an instance of \MSF represented by $\bigO(n \log n)$ evenly distributed edges.

\subsection{Borůvka's algorithm for graphs with $\bigO(n \log n)$ evenly distributed edges}
\label{s:boruvka}
In this subsection we provide an algorithm computing a Minimum Spanning Forest, for graphs given as evenly distributed list of $\bigO(n \log n)$ edges. Proposed procedure, besides finding a Minimum Spanning Forest $F$, provides an additional information for the protocol identifying the $F$-light edges. 

Let us define a \emph{Borůvka's decomposition} in the following way: let $C^i = \set{C^i_1, C^i_2, \dots, C^i_k }$ denotes the set of connected components at the beginning of the $i$th phase of Borůvka's algorithm executed on the graph $G$, and $e^i_j$ denotes the MWOE of $C^i_j$. Then the \emph{Borůvka's decomposition} of a graph $G$ is set of pairs $\langle C^i_j, e^i_j \rangle$, for all $i$ and $j$. Furthermore, we say that the \emph{Borůvka's decomposition is known}, if for all $i$ and $j$, all members of $C^i_j$ know the leader (some distinguished node) of the component $C^i_j$ and the edge $e^i_j$.

This subsection is dedicated to prove the following Theorem:
\begin{theorem}
  \label{t:boruvka}
  Given a weighted graph $H$ given as a evenly distributed list of $\bigO(n \log n)$ edges, it is possible to implement Borůvka's algorithm in $\bigO(\log^2 n)$ rounds of \ncclique, \whp,  in such a way that as a result:
  \begin{itemize}
  \item all edges belonging to the \MSF of $H$ are known by at least one endpoint,
  \item Borůvka's decomposition of $H$ is known.
  \end{itemize}
  
\end{theorem}

\begin{proof}
The algorithm consists of phases. At the beginning of each phase we have the following invariants:
\begin{enumerate}
  \item \label{prop:1} each node knows the leader of its component (ID of some fixed node)
  \item each node knows the list of all MWOEs outgoing from it's component, such that they were used for merging in some previous phase
\end{enumerate}

At the beginning of each phase we identify, for each connected component in the current spanning forest, the MWOE. On the top level, the algorithm identifying MWOEs consists of two steps: in the first we identify for each edge, which nodes are the leaders of components on both endpoints of this edge, in the second we identify a MWOE for each component simultaneously via sorting algorithm from Corollary \ref{c:sorting}. 

The first part is as follows. Each node makes two ordered triplets for each edge: for edge $\set{u,v}$ of weights $w$ the pairs are $\langle u, v, w\rangle$ and $\langle v, u, w\rangle$. Then we sort all those pairs lexicographically. Let consider the first coordinate of those pairs (the other can be done in exactly same way). As a result of sorting, each node has some set of $\bigO(\log n)$ pairs. Let consider a single node, let $x_b$ be the first coordinate of a pair that is the smallest among pairs known to the node, and $x_e$ be the first coordinate of the largest. As long as $x_b \not = x_e$ we distinguish three kinds of pairs:
\begin{enumerate}
  \item \label{p:left} triplets $\langle x,y,w\rangle$, such that $x = x_b$
  \item \label{p:middle} triplets $\langle x,y,w\rangle$, such that $x_b < x < x_e$
  \item \label{p:right} triplets $\langle x,y,w\rangle$, such that $x = x_e$
\end{enumerate}

For all triplets of kind \ref{p:middle}, we can identify ID of a leader of a component that contains a node that is first coordinate of such a pair, in $\bigO(1)$ rounds. A node $v$ that holds a pair $\langle x,y\rangle$ can ask a node $x$ about the ID of a leader. Since the set of triplets was sorted, and $x$ was neither the largest or the smallest among first coordinates, only $v$ asks about the ID of a leader of a component containing $x$. Therefore, the node $x$ receives only one message. Since each node has $\bigO(\log n)$  triplets, each node has to ask only $\bigO(\log n)$ different nodes. Therefore, this communication scheme can be executed in $\bigO(1)$ rounds via Theorem \ref{t:routing}.

Then, we can identify ID of all other first coordinates in $\bigO(\log n / \log \log n)$ rounds. If we consider only triplets such that we did not identified a leader for an endpoint yet, then a node has up to two different first coordinates among all remaining triplets. We can treat nodes with two different first coordinates as two virtual node. Thus, each (possibly virtual) node has all triplets with the same first coordinate. Since for each node $v$ we have a disjoint set of nodes that has to receive the ID of a leader of a component containing $v$, it can be done via Theorem \ref{t:multicast} in $\bigO(\log n / \log \log n)$ rounds.

We can do exactly the same procedure in order to identify the ID of a leader of a component containing the second coordinate for all triplets. Then, we sort all triplets such that first and second coordinates are different (edges between components) by the first and third coordinate. This way, for each component we can identify the lightest outgoing edge.

At this point, we have a set of MWOEs, and we want to use  them in order to perform merging phase. Since the set of MWOE edges is evenly distributed, and there are only $\bigO(n)$ such edges, we can use \EP Connected Components algorithm in order to compute a new division into connected components. Since for each old component its root knows the new component ID, using the multicast protocol we can announce it to all nodes in the network in additional $\bigO(\log n / \log \log n)$ rounds.

Since the Borůvka's algorithm requires $\bigO(\log n)$ phases, and we provided an implementation of a single phase that requires $\bigO(\log n)$ rounds, \whp, the whole algorithm requires $\bigO(\log^2 n)$ rounds, \whp.
\end{proof}

The proposed algorithm finds a minimum spanning forest in such a way that each edge $e$ is known by the node in the component that has $e$ as a MWOE. Furthermore, each node knows up to $\bigO(\log n)$ edges, as there is at most one MWOE for a node, in each phase of the algorithm.

\subsection{Finding the $F$-light edges}
\label{s:f_light_edges}
The main challenge in the \ncclique implementation of the KKT Sampling approach is to provide an efficient procedure finding the set of $F$-light edges. In this section we prove the following Theorem:
\begin{theorem}
\label{t:f_light_edges}
Given a spanning forest $F$ of a random subgraph of $G$, sampled with probability $1/ \log n$ using $k$-wise independent random hash functions, for some $k \geq 9\log n$, such that its Borůvka's decomposition is known, it is possible to find the set of $F$-light edges in $G$, using $\bigO(\log^2 n)$ rounds of \ncclique, \whp.
\end{theorem}

The remaining part of this subsection is a proof of Theorem \ref{t:f_light_edges}. The approach is similar to the one provided by Pemmaraju and Sardeshmukh \cite{Pemmaraju_2016_MST_o(m)}, however we have to provide some adjustments required for the \ncclique model. 

Let $C^i = \set{C^i_1, C^i_2, \dots, C^i_k }$ denotes the set of connected components at the beginning of the $i$th phase of Borůvka's algorithm. Let $e^i_j$ be the MWOE for the component $C^i_j$, if it was used to merge this component. Let $L^i_j$ be the set of the edges outgoing from $C^i_j$, of weight smaller than $w(e^i_j)$. We call $L^i_j$ a set of \emph{light edges} of component $C^i_j$. The proof of Theorem \ref{t:f_light_edges} is divided into two parts:
\begin{itemize}
\item in the first part we discuss some properties of sets $L^i_j$
\item in the second part we show how to use them together with sparse recovery techniques in order to compute $F$-light edges
\end{itemize}
\paragraph{Properties of \emph{light edges}}
Firstly we recall the result of Pemmaraju and Sardeshmukh \cite{Pemmaraju_2016_MST_o(m)}: we can find the set of $F$-light edges by identifying the edges in the set $L$. More precisely, they provide the following Lemma:
\begin{lemma} \cite{Pemmaraju_2016_MST_o(m)}
\label{l:f_light_edges}
The set $L=\bigcup\limits_i\bigcup\limits_jL^i_j$ contains all $F$-light edges.
\end{lemma}

Now, we provide that all those sets are sufficiently small in order to use $\bigO(\log n)$-\emph{edge recovery} technique in order to identify them in $\bigO(\log^2 n)$ rounds. More precisely, we show the following:
\begin{lemma}
\label{l:small}
$\prob{\exists i,j \ .\ |L^i_j| \geq 24\log n} \leq 1 / \Theta(n^2)$.
\end{lemma}

\begin{proof}
  The proof Lemma \ref{l:small}, similarly as Pemmaraju and Sardeshmukh \cite{Pemmaraju_2016_MST_o(m)}, we use the concentration bounds for $k$-wise independent random variables, by Schmidt et al. \cite{Schmidt_1995_dependent_chernoff}, and Pettie and Ramachandran \cite{Pettie_2008_dependent_expected_value}.

  \begin{lemma} \cite{Schmidt_1995_dependent_chernoff}
  \label{l:chernoff}
  Let $X_1 , X_2 , \dots, , X_n$ be a sequence of random $k$-wise independent $0-1$ random variables with $X = \sum \limits_{i=1}^n X_i$. If $k \geq 2$ is even and $C \geq \expv{X}$ then:
  $$ \prob{|X - \expv{X}| \geq T} \leq \left[\sqrt{2}\ \mathsf{cosh}\left(\sqrt{k^3/36C} \right) \right]\cdot \left( \frac{kC}{eT^2} \right)^{k/2} $$ 
  \end{lemma}  

  More precisely we use Lemma \ref{l:chernoff} in the following special case.

  \begin{corollary}
  \label{c:chernoff_2}
  Let $X_1 , X_2 , \dots, , X_n$ be a sequence of random $9 \log n$-wise independent $0-1$ random variables with $X = \sum \limits_{i=1}^n X_i$, such that $\expv{X}<12 \log n$ 
  $$\prob{|X - \expv{X}| > 12\log n} < 1 / \Theta(n^3) $$
  \end{corollary}
  \begin{proof}
  We can apply Lemma \ref{l:chernoff} with $T = C =  12 \log n$,

  \noindent$\prob{|L_i| - \expv{L_i}| \geq 12\log n} \leq $
    
  \noindent$\leq \left[\sqrt{2}\ \mathsf{cosh}\left(\sqrt{(9\log n)^3/ (36 \cdot 12\log n)} \right) \right]\cdot \left( \frac{12\log n \cdot 9\log n}{e \cdot 144\log^2 n} \right)^{\frac{9\log n}{2}} =$

  \noindent$=\left[\sqrt{2}\ \mathsf{cosh} \left( \frac{9\log n} {4\sqrt{3}} \right) \right] \cdot \left(\frac{3}{4e} \right)^{\frac{9\log n}{2}} \leq \left[\frac{1}{\sqrt{2}} \mathsf{cosh}\left(\frac{9\log n}{4\sqrt{3}}\right)\right] \cdot e^{-\frac{9\log n}{2}} =$

  \noindent$=\frac{1} {\sqrt{2}} \left(e^{\frac{9\log n}{4\sqrt{3}}} + e^{-\frac{9\log n}{4\sqrt{3}}} \right) \cdot \left( e^{-\frac{9\log n}{2}} \right) = \frac{1}{\sqrt{2}} \left( e^{\frac{-36\sqrt{3}+9}{4\sqrt{3}}\log n} + e^{\frac{-36\sqrt{3}-9}{4\sqrt{3}}\log n}\right) < \frac{1}{ \Theta(n^3)}$

  \end{proof}

  \begin{lemma}\cite{Pettie_2008_dependent_expected_value}
  \label{l:expected}
   Let $\chi$ be a set of $n$ totally ordered elements and $\chi_p$ be a subset of $\chi$, derived by sampling each element with probability $p$ using a $k$-wise independent sampler. Let $Z$ be the number of unsampled elements less than the smallest element in $\chi_p$. Then $\expv{Z} \leq p^{-1} (8(\pi / e)^2 + 1)\text{ for }k \geq 4$.
  \end{lemma}
  
  By Lemma \ref{l:expected}, we have that the expected number of the edges in $L^i_j$ is no larger than $(8(\pi / e)^2 + 1) \log n < 12 \log n$. Thus, we can use Corollary \ref{c:chernoff_2} to get that the probability of $|L^i_j| > 12 \log n$ for some $i,j$ is smaller than $1 / \Theta(n^3)$. Again, by the union bound we have that $\prob{\exists i,j \ .\ |L^i_j| \geq 24\log n} \leq 1 / \Theta(n^2)$, which proves Lemma \ref{l:small} correct.
\end{proof}

\paragraph{Detection of the $F$-light edges via sparse recovery}
Since each set $L^i_j$ contains only edges lighter than the MWOE of $C^i_j$. By the Lemma \ref{l:small}, $|L^i_j| \in \bigO(\log n)$. Moreover, since \emph{Borůvka's decomposition} of $F$ \emph{is known}, each node for each phase knows
\begin{itemize}
\item the MWOEs of its component 
\item ID of a leader of its component 
\end{itemize}
Since for a fixed value of $i$ all sets $C^i_j$ are disjoint, by Theorem \ref{t:k_recovery} we can recover the sets $L^i_j$ for all $j$ and some fixed $i$ simultaneously in $\bigO(\log n)$ rounds, \whp. Since our implementation of Borůvka's algorithm has $\bigO(\log n)$ phases, we can identify all the edges in $L = \bigcup\limits_i\bigcup\limits_j L^i_j$ in $\bigO(\log^2 n)$ rounds, \whp. By Lemma \ref{l:f_light_edges} $L$ contains all $F$-light edges, therefore Theorem \ref{t:f_light_edges} is correct.

\subsection{\MSF for $F$ extended by $F$-light edges}
At this point the representation of $F$ extended by the $F$-light edges is almost evenly distributed. Since $F$ was computed by sequence of $\bigO(\log n)$ simulations of \EP \SF algorithm, each node has at most $\bigO(\log n) \cdot \bigO(\log n) = \bigO(\log^2 n)$ edges of $F$. Moreover, each node knows up to $\bigO(\log n)$ $F$-light edges per single phase of Borůvka's algorithm, hence $\bigO(\log n) \cdot \bigO(\log n) = \bigO(\log^2 n)$ $F$-light edges in total. Therefore, we have the following problem -- total number of edges is $\bigO(n) + \bigO(n) \cdot \bigO(\log n) = \bigO(n \log n)$, and the edges are distributed in such a way that each node holds $\bigO(\log^2 n)$ of them. Our goal is to execute some load balancing procedure, such that after the execution all edges are evenly distributed among the nodes (each node holds $\bigO(\log n)$ edges).

To do so, we use virtual nodes and the approach presented while introducing communication protocols. In particular we can assign to a node holding $k$ keys, $\Theta(k / \log n)$ virtual \emph{helper} nodes. It can be done in $\bigO(\log n / \log \log n)$ rounds, and as a result, each node knows the set of IDs of its helper nodes. Then, each node can evenly partition the keys (edges) stored in the local memory, and send $\Theta(\log n)$ keys to each helper node. As we have guarantee that each single node knows $\bigO(\log^2 n)$ edges, each node requires only $\bigO(\log n)$ helper nodes. Therefore, it can send a message to all of them in $\bigO(1)$ rounds. Thus, we can simply send the keys to the helper nodes one by one, in $\bigO(\log n)$ rounds.

Since there are $\bigO(n \log n)$ edges in total, and each helper node holds $\Theta(\log n)$ keys, we have only $\bigO(n)$ virtual nodes. 

\subsection{Summary of the Minimum Spanning Forest algorithms}
In Section \ref{s:mst}, we provided an analysis of the Spanning Forest and Minimum Spanning Forest algorithms relying on a recursive application of the KKT Sampling approach, with probability of sampling $p = 1 / \log n$, using $9\log n$-wise independent random variables. In particular we proved that it is possible to:
\begin{itemize}
\item identify a minimum spanning forest and Borůvka's decomposition of a graph with $\bigO(n \log n)$ evenly distributed edges in $\bigO(\log^2 n)$ rounds, \whp. (subsection \ref{s:boruvka})
\item identify a set of $F$ light edges on each level of recursion, in $\bigO(\log^2 n)$ rounds, \whp (subsection \ref{s:f_light_edges})
\item reduce a problem of finding a minimum spanning forest of $F$ extended by the set of $F$-light edges to a problem of finding a minimum spanning forest on a graph with $\bigO(n \log n)$ evenly distributed edges.
\end{itemize}

Therefore, we can execute whole algorithm in $\bigO\paren{\ceil{\frac{\log \Delta}{\log \log n}}} \cdot \bigO(\log^2 n)$ rounds, which proves Theorem \ref{t:mst} correct.
\end{proof}

\section{Conclusions}
Even though the input size is significantly larger than the size of communication allowed in a single round, we still can apply some results from the \MPC models to the \ncclique. The key part is to provide a reduction of the problem on general graph into some subproblems of size comparable to the size of communication. 

In the case of the Minimum Spanning Tree, the KKT Sampling technique turned out to generate subproblems with a structure that allowed an efficient implementation in the \ncclique model. It remains an open question whether any other filtering / sampling techniques that proved to be efficient for some different problems, can be implemented in the \ncclique model.

\bibliographystyle{abbrv} 
\bibliography{ref}

\end{document}